\documentclass{article}
\usepackage{amsmath}
\usepackage{amssymb}
\usepackage{amsfonts}
\usepackage{a4wide}
\usepackage{amsthm}

\def\mod{\text{ mod }}
\def\m{{\frak m}}

\newcounter{NN}
\newtheorem{theorem}[NN]{Theorem}
\newtheorem{corollary}[NN]{Corollary}
\newtheorem{conjecture}[NN]{Conjecture}
\newtheorem{lemma}[NN]{Lemma}

\def\m{{\frak m}}

\begin{document}
\bibliographystyle{plain}
\title{From integrable equations to Laurent recurrences}
\author{Khaled Hamad and Peter H.~van der Kamp}
\date{Department of Mathematics and Statistics\\
La Trobe University, Victoria 3086, Australia\\[5mm]
\today
}

\maketitle

\begin{abstract}
Based on a recursive factorisation technique, first introduced in \cite{vdK}, we show how integrable difference equations give rise to recurrences which possess the Laurent property. We derive non-autonomous Somos-$k$ sequences, with $k=4,5$, whose coefficients are periodic functions with period 8 for $k=4$, and period 7 for $k=5$, and which possess the Laurent property. We also apply our method to the DTKQ-$N$ equation \cite{DTKQ}, with $N=2,3$, and derive Laurent recurrences with $N+2$ terms, of order $N+3$. In the case $N=3$ the recurrence has periodic coefficients with period 8. We demonstrate that recursive factorisation also provides a proof of the Laurent property.
\end{abstract}

\section{Introduction}
A sequence $\{u_n\}_{n=1}^\infty$ defined by $N$ initial values $\{u_n\}_{n=1}^N$ and an $N$th order nonlinear rational recursion,
\begin{equation} \label{nrr}
u_{n+N}=R(u_n,u_{n+1},\ldots,u_{n+N-1}),
\end{equation}
where $R$ is a rational function, is said to have the {\em Laurent property} if, for all $n$, $u_n$ is polynomial in the variables $\{u_n^{\pm 1}\}_{n=1}^N$. The property was first introduced by Hickerson to prove the integrality of a sequence called Somos-6, cf. \cite{Rob}. Indeed, as an immediate consequence of the Laurent property it follows that the sequence obtained by taking $\{u_n=1\}_{n=1}^N$ is an integer sequence, or, a sequence of polynomials if the rational function $R$ depends on additional parameters. For example, with mentioned initial values the (generalized) Somos-4 recurrence
\begin{equation} \label{S4}
\tau_{n+2} \tau_{n-2}=\alpha \tau_{n+1} \tau_{n-1}+\beta \tau_n^2
\end{equation}
provides a sequence of polynomials in two variables $\alpha,\beta$.

Equation (\ref{S4}) was derived (in 1982) by Michael Somos as an addition formula for elliptic functions. It is the prototype Laurent recurrence, and it has many beautiful properties. The sequence of numbers that one gets by taking $\alpha=\beta=1$
is referred to as {\em the} Somos-4 sequence. Its integrality (and of related sequences) was a great mystery initially \cite{Cra,Gal,Mal,Som}. Robinson showed that the $i$-th and $j$-th terms of the Somos-4 sequence are relatively prime whenever $|i-j|\leq4$, and he infered that that for any given $m\in\mathbb N$ the sequence modulo $m$ is periodic \cite{Rob}. Everest et al. \cite{EMW} showed that every term beyond the fourth has a primitive divisor, i.e. a prime which does not divide any preceding term. Kanki et al. \cite{KMMT}
have proven that all terms of Somos-4 are irreducible Laurent polynomials in their initial values and pairwise co-prime, as Laurent polynomials. A seemingly unnoticed divisibility property for the Somos-4 polynomials was recently found by one of the authors \cite{vdK}. A so called near-addition formula has been proven in \cite{Ma}.
Somos-4 is closely connected to an elliptic divisibility sequence (EDS) \cite{Hone5,HS,PooSwa,War}, the theory of which recently found application in cryptography \cite{Shi}, and in generating large primes \cite{EMS}. An explicit solution for $\tau_n$ in terms of the Weierstrass elliptic function can be found in \cite{Hone4,Hone5}. From an integrable systems viewpoint, the Somos-4 recurrence arises as a bilinearisation of the following QRT map \cite{QRT},
\begin{equation} \label{fqrt}
f_{n+1} f_n f_{n-1}=\alpha +\frac{\beta}{f_n},
\end{equation}
through the relation
\begin{equation} \label{ftau}
f_n=\frac{\tau_{n+1}\tau_{n-1}}{\tau_n^2},
\end{equation}
which encodes the singular confinement of the QRT map \cite{HoneLP,HoneSC}. Furthermore, it is a special case of the Hirota-Miwa equation \cite{GR,Mas,WTS}.

A deeper understanding of the Laurent property, for a wide class of recurrences, came with the work of Fomin and Zelevinski \cite{FoZeC,FoZe}, and subsequent development \cite{ACH,FM,LP}. The algebraic combinatorial setting of cluster algebras has had profound impact in diverse areas of mathematics, such as algebraic Lie theory \cite{GLS}, Poisson geometry \cite{GSV}, higher Teichm\"uller theory \cite{FG}, the representation theory of quivers and finite-dimensional algebras \cite{BMRRT}, and integrable systems \cite{FoHo}, cf. the cluster algebra portal \cite{Fom}.

In this paper we describe a method to obtain recurrences which possess the Laurent property, such as (\ref{S4}), from equations that are integrable, such as (\ref{fqrt}), which is different than via a transformation such as (\ref{ftau}). Our method is based on a recursive factorisation technique, which was used in \cite{vdK} to establish polynomial upper bounds on the growth of degrees of rational mappings. Slow growth (= low complexity = vanishing algebraic entropy) is a good indicator (better than singular confinement) of the integrability of a mapping \cite{A,BMR,BV,FV,OTGR,Ves}.

Given a rational recurrence (\ref{nrr}) we homogenise $u_n=a_n/b_n$,
which yields a system of recurrences for polynomial sequences $\{a_n\}_{n=1}^\infty$
and $\{b_n\}_{n=1}^\infty$. Such a system has two ultra-discrete limits; one gives an upper bound on the growth of the degrees of the polynomials $a_n$ and $b_n$, and the other
yields a lower bound on the multiplicities of their divisors. The degree of $u_n$ is
obtained from the degree of $a_n$ (or $b_n$) minus the degree of the greatest common divisor $g_n=gcd(a_n,b_n)$. Thus, one has to control the divisors of $a_n$ and $b_n$. By iterating the
system finitely many times (one has to iterate $N$ times, the order of the recurrence) and using the observed factorisation as initial values in the ultra-discrete
system for multiplicities, one obtains a lower bound on the multiplicities of divisors
\cite{Ham,vdK}. In many cases this lower bound on the multiplicities is sharp. In any case, by recursively defining the next divisor to be the quotient of a term in the sequence after division by the previous divisors, one produces an exact factorisation of the polynomial sequences (although not necessarily into irreducible factors)\footnote{If one is after degree growth one now writes the degree of $a_n$ (or $b_n$) as a convolution of the degrees of the divisors and their multiplicities. Using (the solution to) the degree
recurrence one may then obtain a recursion for the degrees of the divisors and, when all but finitely many
divisors are common, retrieve an upper bound on the growth of degrees of $u_n$ \cite{Ham,vdK}.}.
Substitution of these factorisations, when all but finitely many divisors are common, into
the system of recurrences for our polynomial sequences yields a nonlinear rational recurrence
for the divisors. Clearly, by definition, the divisors are polynomial. Hence one expects the recurrence to
possess the Laurent property.

If one starts with a recurrence (\ref{nrr}) that has the Laurent property, what will happen is
that all divisors but powers of the initial variables, will be common to both $a_n$ and $b_n$ for
all $n$. This then proves the Laurent property.

In section \ref{SQRT} we show how the QRT-map (\ref{fqrt}), via recursive factorisation, gives rise
to a Somos-4 recurrence of the form (\ref{S4}) but where the coefficients are now functions of the initial values of the QRT-map, $\alpha=\alpha_n(f_1,f_2), \beta=\beta_n(f_1,f_2)$, which
satisfy the periodicity conditions $(\alpha,\beta)_{n+p}=(\alpha,\beta)_n$ with $p=8$.
Similarly we show how another QRT-map yields a Somos-5 recurrence where the coefficients are periodic functions with period $p=7$. In section \ref{SSOM} we follow the same procedure starting with the Somos-4,5 sequences themselves. Surprisingly, or not, they give rise to Somos-4 and Somos-5 recurrences with more general periodic coefficients than those obtained in section \ref{SQRT}. Explicitly we have found
\begin{equation} \label{S4P}
c_{n+2}c_{n-2}=\alpha_{n} c_{n+1}c_{n-1} + \beta_{n} c_{n}^{2},
\end{equation}
with coefficients
\begin{equation}
\alpha_{n}=\alpha \prod_{i=1}^{4} \tau_{i}^{p_{n-i \mod 8}}, \qquad 
\beta_{n}=\beta \prod_{i=1}^{4} \tau_{i}^{q_{n-i \mod 8}}.
\end{equation}
where $p=[1,0,0,1,0,0,1,0]$\footnote{We have $p_1=1$.}, $q=[0,0,1,0,1,0,0,2]$, and
\begin{equation} \label{S5P}
d_{n+3}d_{n-2}=\gamma_{n} d_{n+2} d_{n-1}+\delta_{n} d_{n} d_{n+1},
\end{equation}
with coefficients
\begin{equation}
\gamma_{n}=\gamma \prod_{i=1}^{5} \sigma_{i}^{r_{n-i \mod 7}}, \qquad 
\delta_{n}=\delta \prod_{i=1}^{5} \sigma_{i}^{s_{n-i \mod 7}}.
\end{equation}
where $r=[1,0,0,0,1,0,0]$, $s=[0,0,1,0,0,1,1]$. Both equations (\ref{S4P},\ref{S5P}) are
special cases of the Hirota-Miwa equation
\begin{equation} \label{HWE}
h_nh_{n+w}=\alpha_nf_{n+v_1}f_{h+u_1}+\beta_nf_{n+v_2}f_{h+u_2}
\end{equation}
whose integrability condition
\begin{equation} \label{HWC}
\alpha_{n}\alpha_{n+w}\beta_{n+v_1}\beta_{n+u_1}
=\alpha_{n+v_2}\alpha_{n+u_2}\beta_{n}\beta_{n+w}
\end{equation}
is equivalent to Laurentness, see \cite{Mas}. Condition (\ref{HWC}) is satisfied for both equations
(\ref{S4P}) and (\ref{S5P}).

In the final section we apply our method to two equations quite unrelated to anything Somos. We consider the first two members of the hierarchy of equations
\begin{equation}
(\sum_{k=0}^{N}u_{n+k})(\prod_{l=1}^{N-1}u_{n+l})=\phi. 
\end{equation}
which was introduced in \cite{DTKQ}, and whose degree growth has been studied in \cite{Ham}. For $N=2$ the map is another QRT-map, which we relate to the fifth order Laurent recurrence
\begin{equation}
e_{n+5}e_{n+2}^{2}e_{n+1}
= \phi e_{n+3}^{2}e_{n+2}^{2}
- e_{n+4}e_{n+3}^{2}e_{n} 
- e_{n+4}^{2}e_{n+1}^{2}.
\end{equation}
For $N=3$ we find that the ultra-discretisation of the homogenised system does not describe
the multiplicities of the second divisor. Using a number of primes as initial values enables us to
iterate the system sufficiently many times and so to formulate a conjecture for these multiplicities.
Via recursive factorisation we arrive at the following sixth order periodic Laurent recurrence,
\begin{align*}
\epsilon_{n+3}c_{n-1}c_{n-2}c_{n}^{2}c_{n+3}
=&\frac{\alpha}{\epsilon_{n+1}\epsilon_{n+2}}c_{n-1}c_{n}^{3}c_{n+1}
-\epsilon_{n}c_{n+2}c_{n-3}c_{n+1}c_{n}^2 \\
&-\epsilon_{n+1}c_{n-2}^{2}c_{n+2}c_{n+1}^{2}
-\epsilon_{n+2}c_{n-2}c_{n-1}^2c_{n+2}^2,
\end{align*}
with
$\epsilon_n=u_{2}^{\zeta_{n \mod 8}}$ and $\zeta=[0,1,0,-1,-1,2,-1,-1]$.

If one starts with an integrable equation and obtains via recursive factorisation a Laurent
recurrence for the divisors of the numerators and denominators, the divisors will depend on
both the parameters and the initial values of the integrable equation. For the periodic Somos
sequences this dependence is realised in the coefficients from the Laurent recurrence, and
we can start the recurrence with unit initial values. Thus the polynomiality of the divisors is
completely explained by the Laurentness of the recurrence. For the recurrences we have obtained
for the DTKQ equations this is not the case. Here we have to initialise the recurrences with
initial values that depend in a specific way on the initial values of the DTKQ equation. Therefore
in these cases the Laurentness of the recurrences is not enough to explain the polynomiality
of the divisors. One would need a strong Laurent property such as given for Somos-4,5 in \cite{HS}.
This issue is left open for future research.

\section{From QRT maps to Somos-4 \& 5 recurrences with periodic coefficients} \label{SQRT}
In this section we show how by homogenisation, an ultra-discrete limit and recursive factorisation the QRT-map (\ref{fqrt}) leads to a special case of
periodic Somos-4, equation (\ref{S4P}). A similar result for Somos-5 is also given.

\subsection{To periodic Somos-4}
We substitute $f_{n}=a_{n}/b_{n}$ in (\ref{fqrt}). This gives
\[
\frac{a_{n+1}}{b_{n+1}}
=
\frac{w_{n+1}b_{n}b_{n-1}}{a_{n-1}a_{n}^{2}},
\]
with $w_{n+1}:=\alpha a_{n}+\beta b_{n}$, from which we obtain
a system of recurrences for polynomial sequences $\{a_n\}$ and $\{b_n\}$:
\begin{align}
a_{n+1}&=w_{n+1}b_{n}b_{n-1} \label{map7}  \\
b_{n+1}&={a_{n-1}a_{n}^{2}} \label{map8},
\end{align}
which we supplement with initial values $a_{1}=f_{1}$, $a_{2}=f_{2}$, $b_{1}=b_{2}=1$.
Iterating (\ref{map7}) and (\ref{map8}) three more times give us:
\begin{align*}
a_{n+2} = a_{n-1}a_{n}^2b_{n} r_1,
\ &b_{n+2}=a_{n}b_{n-1}^2b_{n}^2 w_{n+1}^2 , \\
a_{n+3} = a_{n-1}a_{n}^4b_{n-1}^2b_{n}^3r_2 w_{n+1}^2 ,
\ &b_{n+3} = a_{n-1}^2a_{n}^4 b_{n-1} b_{n}^3 w_{n+1}r_{1}^2, \\
a_{n+4} = a_{n-1}^3a_{n}^9b_{n-1}^4b_{n}^8  r_{2}^2r_3w_{n+1}^4,
\ &b_{n+4} =  a_{n-1}^3 a_{n}^{10}b_{n-1}^4b_{n}^7 r_{2}^{2}w_{n+1}^{4},
\end{align*}
where $\{r_i\}_{i=1}^{3}$, are irreducible polynomials in $a_{n-1}$, $b_{n-1}$, $a_n$, $b_n$, $\alpha$ and $\beta$. We observe the following factorisation properties: $w_{n+1}$ does not divide $a_{n+2}$, it divides $b_{n+2}$ and $a_{n+3}$ with multiplicity 2, it divides $b_{n+3}$ with multiplicity 1, and $w_{n+1}$ is a divisor of both $a_{n+4}$ and $b_{n+4}$ with multiplicity 4. Furthermore, from (\ref{map7}) and (\ref{map8}), we find the following ultra-discrete system of recurrences for multiplicities:
\begin{align*}
 m_{n+2}^{a} &\geq min(m_{n+1}^{a} +m_{n}^{b} + m_{n+1}^{b}; 2 m_{n+1}^{b} +m_{n}^{b}),\\
m_{n+2}^{b} &=m_{n}^{a} +2m_{n+ 1}^{a}.
\end{align*} where $m_{i}^{p}(f)$ denotes the multiplicity of a polynomial $f$ in polynomial $p_{i}$ and we have suppressed the dependence on $f$. 
Using the equal sign in the first equation, we will get a lower bound for the multiplicities, which we denote using Euler's frakture typesetting. Thus, we will employ the following system:
\begin{align}\label{map9}
\m_{n+2}^{a} &= min(\m_{n+1}^{a} +\m_{n}^{b} + \m_{n+1}^{b}; 2 \m_{n+1}^{b} +\m_{n}^{b}),\\
\m_{n+2}^{b} &=\m_{n}^{a} +2\m_{n+ 1}^{a}.\label{map10}
\end{align}
To get a lower bound for the multiplicity of $w_{k}$ $(k>2)$ in the sequences $\{a_{n}\}$ and $\{b_{n}\}$, we solve (\ref{map9}) and (\ref{map10}) with the following initial values: $\m_{k+1}^{a}=0$, $\m_{k+1}^{b}=2$, $\m_{k+2}^{a}=2$, $\m_{k+2}^{b}=1$ and $\m_{k+3}^{a}=\m_{k+3}^{b}=4$. We find, for $n \geqslant k+3$, that $\m_{n}^{a}(w_{k})= \m_{n}^{b}(w_{k})=\m_{n-k}$, where
\[
\m_{1}=0,\ \m_{2}=2,\ \m_{n+1}=2\m_{n}+\m_{n-1}.
\]
This can be seen by taking   $ \m_{k}^{a}=\m_{k}^{b}$ in the right hand sides of (\ref{map9}) and (\ref{map10}). One finds equality and hence $\m_{n+2}^{a}=\m_{n+2}^{b}$. We define sequences $\{\m_{n}^{a}(c_{i})\}_{n=1}^{\infty}$ and $\{\m_{n}^{b}(c_{i})\}_{n=1}^{\infty}$, for $i\in\{1,2\}$, by (\ref{map9}) and (\ref{map10}) and initial values $\m_{j}^{a}(c_{i})=\delta_{ij}$ and $\m_{j}^{b}(c_{i})=0$. 

The polynomials  $a_{n}$ and $b_{n}$ can be expressed  in terms of a sequence $\{c_{k}\}_{n=1}^\infty$, of polynomials in $a_1$, $a_2$, $\alpha$ and
$\beta$. Each polynomial $c_{n}$ is defined as the quotient of $a_{n}$ after division by powers of $c_{i < n}$, as follows, 
\begin{equation} \label{map14}
a_{n}=
\left\{
\begin{array}{ll}
c_{n} &\text{ if } n\leqslant 3, \\
c_{1}c_{2}^{2}c_{4} &\text{ if } n=4, \\
c_{1}^{\m_{n}^{a}(c_{1})}c_{2}^{\m_{n}^{a}(c_{2})}(\prod_{i=3}^{n-3}c_{i}^{\m_{n-i}})c_{n-2}^{2}c_{n} &\text{ if } n > 4.
\end{array}
\right.
\end{equation}
It is clear that $c_{n}$ is polynomial because $c_{i}|w_{i}$ for all $i > 4$ and hence $\m_{n}^{a}(c_{i}) \geqslant \m_{n}^{a}(w_{i}) $. 
We know that $\prod_{i=1}^{n}c_{i}^{\m_{n}^{b}(c_{i})}|b_{n}$. Taking $b_{n}$ to be given by
\begin{align}\label{map15}
b_{n}=
\left\{
\begin{array}{ll}
1 &\text{ if } n\leqslant 2, \\
c_{n-2}c_{n-1}^{2} &\text{ if } n \in \{3,4\}, \\
c_{1}^{\m_{n}^{b}(c_{1})}c_{2}^{\m_{n}^{b}(c_{2})}(\prod_{i=3}^{n-3}c_{i}^{\m_{n-i}})c_{n-2}c_{n-1}^{2}, &\text{ if } n>4,
\end{array}
\right.
\end{align} we can verify equation (\ref{maba}) is satisfied. Thus, defining $g_n=\gcd(a_n,b_n)$ to be the greatest common divisor of $a_n$ and $b_n$, we get
\begin{align}\label{map16}
g_{n}=\prod_{i=1}^{n}c_{i}^{\m_{n}^{g}(c_{i})}=c_{1}^{\m_{n}^{g}(c_{1})}c_{2}^{\m_{n}^{g}(c_{2})}(\prod_{i=3}^{n-3}c_{i}^{\m_{n-i}})c_{n-2},
\end{align} where $\m_{n}^{g}(c_{i})=min(\m_{n}^{a}(c_{i}),\m_{n}^{b}(c_{i}))$. Note, from
$\frac{b_{n}}{g_{n}}=c_{1}^{{\m_{n}^{b} (c_{1})} -{\m_{n}^{g}(c_{1})} }  c_{2}^{{\m_{n}^{b}(c_{2})}-{\m_{n}^{g}(c_{2})}}c_{n-1}^{2}$, it can be seen that the map (\ref{fqrt}) does not posses the Laurent property.

We now consider the lower bounds for the multiplicities of $c_{1}$, $c_{2}$ in $a_n$ and $b_n$ and observe the following differences are periodic.
\begin{lemma}
We have:
\[
\m_{k}^{a}(c_{i})-\m_{k}^{b}(c_{i})=
\left\{
\begin{array}{ll}
v_{k \mod  8} &\text{ if } i=1, \\
v_{k-3 \mod  8} &\text{ if } i=2,
\end{array}
\right.
\] where $v=[1,0,-1,1,-1,0,1,-2]$.
\end{lemma}
\begin{proof}
By induction. We only give the case $k \equiv 1 \mod 8$ for $c_{1}$. For $k=1$, it is evident. Suppose the claim is true for $i < k$. Then, we have  $\m_{k-1}^{a}(c_{1})=\m_{k-1}^{b}(c_{1})-2$ and   $\m_{k-2}^{a}(c_{1})=\m_{k-2}^{b}(c_{1})+1$, hence
  \begin{align*}
 \m_{k}^{a}(c_{1})&= min(\m_{k-1}^{a}(c_{1}) + \m_{k-2}^{b}(c_{1}) + \m_{k-1}^{b}(c_{1}), 2\m_{k-1}^{b}(c_{1}) + \m_{k-2}^{b}(c_{1}))\\\notag
 &= 2 \m_{k-1}^{b}(c_{1}) + \m_{k-2}^{b}(c_{1})-2, \text{ and } \\
 \m_{k}^{b}(c_{1}) &= \m_{k-2}^{a}(c_{1}) +2\m_{k-1}^{a}(c_{1})\\
 &=2\m_{k-1}^{b}(c_{1}) + \m_{k-2}^{b}(c_{1})-3.
\end{align*} So $\m_{k}^{a}(c_{1})-\m_{k}^{b}(c_{1})=1$. By applying the same technique for other cases  $(k\equiv 2,...,8\mod 8)$ and for $c_{2}$, the lemma is proven.
\end{proof} 
 
From (\ref{map14}), (\ref{map15}), (\ref{map16}), it follows that
\begin{equation}\label{map17}
\alpha_{n}:=\frac{\alpha a_{n}}{c_{n}c_{n-2}g_{n}}\quad \text{ and } \quad
\beta_{n}:=\frac{\beta b_{n}}{c_{n-1}^{2}g_{n}}
\end{equation}
are polynomials in $c_{1}$, $c_{2}$. As a Corollary to Lemma 1, it follows that $\alpha_{n}$ and $\beta_{n}$   are periodic sequences of  period $8$.
\begin{corollary}
We have: 
\begin{equation}\label{map19}
\alpha_{n}=\alpha c_{1}^{p_{n\mod8}}c_{2}^{p_{n-3\mod8}}\quad \text{ and }  \quad \beta_{n}=\beta c_{1}^{q_{n\mod8}}c_{2}^{q_{n-3\mod8}},
\end{equation}
with $p=[1,0,0,1,0,0,1,0] \quad and \quad q=[0,0,1,0,1,0,0,2]$.
\end{corollary}
\begin{proof}
We have:
\begin{equation*}
\alpha_{n}=\frac{\alpha a_{n}}{g_{n}c_{n}c_{n-2}}=\alpha c_{1}^{{\m_{n}^{a} (c_{1})} -{\m_{n}^{g}(c_{1})} }  c_{2}^{{\m_{n}^{a}(c_{2})}   -{\m_{n}^{g}(c_{2})}},
\end{equation*} where
\[
\m_{n}^{a} -\m_{n}^{g}=
\left\{
\begin{array}{ll}
\m_{n}^{a} -\m_{n}^{b}&\text{ if } \m_{n}^{a}-\m_{n}^{b}> 0, \\
0 &\text{ if } \m_{n}^{a} -\m_{n}^{b} \leqslant 0.
\end{array}
\right.
\] Therefore,
\[
\m_{n}^{a} (c_{i}) -\m_{n}^{g}(c_{i})=
\left\{
\begin{array}{ll}
p_{n \mod 8} &\text{ if } i=1, \\
p_{n-3 \mod 8} &\text{ if } i=2,
\end{array}
\right.
\] where $p_{k}=max(0,v_{k})$. Similarly, we have:
\begin{equation*}
\beta_{n}=\frac{\beta b_{n}}{g_{n}c_{n-1}^{2}}=\beta c_{1}^{{\m_{n}^{b} (c_{1})} -{\m_{n}^{g}(c_{1})} }  c_{2}^{{\m_{n}^{b}(c_{2})}-{\m_{n}^{g}(c_{2})}}, 
\end{equation*} where
\[
\m_{n}^{b} (c_{i}) -\m_{n}^{g}(c_{i})=
\left\{
\begin{array}{ll}
q_{n \mod 8} &\text{ if } i=1, \\
q_{n-3 \mod 8} &\text{ if } i=2,
\end{array}
\right.
\] with $q_{k}=max(0,-v_{k})$.
\end{proof}
 
\begin{theorem}
The polynomials $c_n$, as defined by (\ref{map14}), satisfy
\begin{equation}\label{XV1}
\quad c_{3}=\alpha c_{2} +\beta,\quad c_{4}=\alpha c_{3}+\beta c_{1}c_{2}^{2},\quad c_{5}=\alpha c_{1}c_{2}c_{4}+\beta c_{3}^{2},\quad c_{6}=\alpha c_{5}c_{3}+\beta c_{1}c_{4}^{2},
\end{equation} and, for $n \geqslant 6$, 
\begin{equation}\label{X2}
c_{n-3}c_{n+1}=\alpha_{n}c_{n}c_{n-2}+\beta_{n}c_{n-1}^{2}.
\end{equation}
\end{theorem}  
\begin{proof}
Using equations (\ref{map7}) and (\ref{map8}), initial values and  (\ref{map14}), we find:
\begin{align*}
c_{3}=a_{3}=(\alpha a_{2}+\beta b_{2})b_{1}b_{2}=(\alpha c_{2}+\beta).
\end{align*}  Furthermore, 
\begin{align*}
c_{4}=\frac{a_{4}}{c_{1}^{\m_{4}^{a}(c_{1})}c_{2}^{\m_{4}^{a}(c_{2})}}=\frac{(\alpha a_{3}+\beta b_{3})b_{2}b_{3}}{c_{1}c_{2}^{2}}=\alpha c_{3}+\beta c_{1}c_{2}^{2},
\end{align*} as $b_{3}=c_{1}c_{2}^{2}$, $\m_{4}^{a}(c_{1})=1$ and $\m_{4}^{a}(c_{2})=2$. Similarly, the formulae for $c_{5}$ and $c_{6}$ are obtained.
Solving equations (\ref{map17}) for $a_{n}$ and $b_{n}$ and substituting in equation (\ref{map7}), we find:
\begin{align*}
\frac{\alpha_{n+1}c_{n+1}c_{n-1}g_{n+1}}{\alpha}=(\alpha_{n}c_{n}c_{n-2}g_{n}+\beta_{n} c_{n-1}^{2}g_{n}) \frac{\beta_{n-1}\beta_{n}c_{n-1}^{2}c_{n-2}^{2}g_{n-1}g_{n}}{\beta^2}.
\end{align*} Thus,
\begin{align*}
c_{n-3}c_{n+1}=Z_{n}(\alpha_{n}c_{n}c_{n-2}+\beta_{n} c_{n-1}^{2}),
\end{align*} with
\begin{align*}
Z_{n}=\frac{\beta_{n-1}\beta_{n}}{\beta^{2}}\frac{\alpha}{\alpha_{n+1}}\frac{g_{n-1}g_{n}^{2}}{g_{n+1}} c_{n-1}c_{n-2}^{2}c_{n-3}.
\end{align*} Substituting in equation (\ref{map8}) gives us:
\begin{align*}
g_{n+1}=\frac{\beta}{\beta_{n+1}}  \frac{\alpha_{n}^{2}\alpha_{n-1}}{\alpha^{3}}g_{n-1}g_{n}^{2}c_{n-1}c_{n-2}^{2}c_{n-3}, 
\end{align*} which we use to simplify 
\begin{align*}
Z_{n}=\frac{\alpha^{4}\beta_{n-1}\beta_{n}\beta_{n+1}}{ \alpha_{n-1}\alpha_{n}^{2}\alpha_{n+1}\beta^{3}}=1
\end{align*} as $q_{n-1}+q_{n}+q_{n+1}=p_{n-1}+2p_{n}+p_{n+1}$.
\end{proof} The fact that the sequence $\{c_{n}\}_{n=1}^\infty$, with special initial values given by (\ref{XV1}) and generated by the rational recurrence (\ref{X2}), is a polynomial sequence is curious. First of all, it follows from the definition of $c_{n}$ given by (\ref{map14}) which is based on factorization properties of the QRT map (\ref{fqrt}). But there is a second explanation. When we express the coefficients, cf. Lemma 2, in terms of the initial values of the QRT-map (\ref{fqrt}), $c_{1}=a_{1}=f_{1}$ and $c_{2}=a_{2}=f_{2}$, i.e.
\begin{equation}\label{M3}
\alpha_{n}^{f}=
\left\{
\begin{array}{ll}
\alpha f_{1} &\text{ if } n \equiv 1 \mod 8, \\
\alpha f_{2} &\text{ if } n \equiv 2 \mod 8,\\
\alpha f_{1}f_{2} &\text{ if } n \equiv 4,7 \mod 8,\\
\alpha &\text{ if } n \equiv 3,5,6,8 \mod 8,
\end{array}
\right. \quad
\beta_{n}^{f}=
\left\{
\begin{array}{ll}
\beta &\text{ if } n \equiv 1,2,4,7 \mod 8, \\
\beta f_{1}f_{2}^2 &\text{ if } n \equiv 3 \mod 8,\\
\beta f_{1} &\text{ if } n \equiv 5 \mod 8,\\
\beta f_{2} &\text{ if } n \equiv 6 \mod 8,\\
\beta f_{1}^{2}f_{2} &\text{ if } n \equiv 8 \mod 8,
\end{array}
\right.
\end{equation}
and supplement the recurrence
\begin{equation}\label{X2X}
c_{n-3}c_{n+1}=\alpha^f_{n}c_{n}c_{n-2}+\beta^f_{n}c_{n-1}^{2}.
\end{equation}
with initial values $c_{i}=1$ for $i \in \{-1,0,1,2\}$ we find the following expressions 
\[
\begin{array}{ll}
c_{3}=\alpha f_{2}+\beta, 
\quad c_{4}=\alpha c_{3}+\beta f_{1}f_{2}^{2},\quad 
c_{5}=\alpha f_{1}f_{2}c_{4}+\beta_{3}^{2},\quad  
c_{6}=\alpha c_{5}c_{3}+\beta f_{1}c_{4}^{2},
\end{array}
\]
which agree with (\ref{XV1}). Therefore, the fact that the sequence consist of polynomials can also be explained by the Laurent property of (\ref{X2X}), which we establish in section \ref{subs}.

\subsection{To periodic Somos-5}
We will now show how the QRT-map
\begin{equation} \label{hqrt}
h_{n+1} h_{n} h_{n-1}=\gamma h_n + \delta,
\end{equation}
which is related to Somos-5,
\begin{equation} \label{Som5}
\sigma_{n+3}\sigma_{n-2}=\gamma\sigma_{n+2}\sigma_{n-1}+\delta  \sigma_{n+1}\sigma_n,
\end{equation}
via the transformation, cf. \cite{Hone5},
\begin{equation} \label{hs}
h_n=\frac{\sigma_{n+2}\sigma_{n-1}}{\sigma_{n+1}\sigma_{n}},
\end{equation}
leads to a special case of periodic Somos-5, equation (\ref{S5P}).
Substituting
$h_{n}=a_{n}/b_{n}$, the homogenised system for numerators and denominators is given by:
\begin{align} \label{map22}
a_{n+1}&=w_{n+1}b_{n-1} \\
b_{n+1}&=a_{n}a_{n-1},\label{map23}
\end{align} where $w_{n+1}:=\gamma a_{n}+\delta b_{n}$. We take  $\{b_i=1\}_{i=1}^2$, so that $\{a_i=h_i\}_{i=1}^2$. Iterating (\ref{map22}), (\ref{map23}), we find:
\begin{align*}
a_{n+2} = b_{n}s_{1},
\ &b_{n+2}=w_{n+1}b_{n-1}a_n, \\
a_{n+3} =a_{n}a_{n-1}s_2,
\ &b_{n+3} = b_{n}b_{n-1}s_{1}w_{n+1}, \\
a_{n+4} =a_{n}b_{n-1}s_3w_{n+1},
\ &b_{n+4} = a_{n} a_{n-1}b_{n} s_{1}s_{2}, \\
a_{n+5} =a_{n} b_{n-1}b_{n}s_{1}s_{4}w_{n+1},
\ &b_{n+5} = a_{n-1} a_{n}^2 b_{n-1}s_{2}s_{3}w_{n+1},
\end{align*} where $\{s_{i}\}_{1}^{4}$ are irreducible polynomials in $\{a_{n+i},b_{n+i}\}_{i=-1}^{0}$, $\delta$, and $\gamma$. In addition, from (\ref{map22}) and (\ref{map23}), the ultra-discrete system of recurrences for a lower bound on the multiplicities is:
\begin{align}\label{map24}
 \m_{n+1}^{a} &= min(\m_{n}^{a} +\m_{n-1}^{b};\m_{n}^{b} +\m_{n-1}^{b}),\\
\m_{n+1}^{b} &=\m_{n}^{a} +\m_{n-1}^{a}.\label{map25}
\end{align} To get a lower bound for the multiplicity of $w_{k}$ $(k>3)$ in the sequences $\{a_{n}\}$ and $\{b_{n}\}$, we solve (\ref{map24}) and (\ref{map25}) with the following initial values: 
$\m_{k+1}^{a}=\m_{k+2}^{a}=\m_{k+3}^{b}=0$ and $\m_{k+1}^{b}=\m_{k+2}^{b}=\m_{k+3}^{a}=\m_{k+4}^{a}=\m_{k+4}^{b}=1$. We find, for all $n \geqslant k+4$, that
$
\m_{n}^{a}(w_{k})= \m_{n}^{b}(w_{k})=\m_{n-k-3}$ where
\[
\m_{1}=1,\  \m_{2}=2,\ \m_{n+2}=\m_{n+1}+\m_{n}.
\]
For $i\in\{1,2\}$ we define sequences $\m_{n}^{a}(d_{i})$ and $\m_{n}^{b}(d_{i})$ by (\ref{map24}) and (\ref{map25}) and  the initial values $\m_{j}^{a}(d_{i})=\delta_{ij}$ and $\m_{j}^{b}(d_{i})=0$. Then, a polynomial sequence $\{d_{n}\}_{n=1}^\infty$ is defined by 
\begin{equation}\label{X4}
a_{n}=
\left\{
\begin{array}{ll}
d_{n} &\text{ if } n\leqslant 4, \\
d_{1}d_{2}d_{5} &\text{ if } n=5,\\
(\prod_{i=1}^{2}d_{i}^{\m^{a}_{n}(d_{i})})(\prod_{i=3}^{n-4}d_{i}^{\m_{n-i-3}})d_{n-3}d_{n}, &\text{ if }n > 5,
\end{array}
\right.
\end{equation} and we have
\begin{equation}\label{map37}
b_{n}=
\left\{
\begin{array}{ll}
1 &\text{ if } n\leqslant 2,\\
d_{n-2}d_{n-1} &\text{ if } n \in \{3,4\},\\
(\prod_{i=1}^{2}d_{i}^{\m^{b}_{i}(d_{i})})(\prod_{i=3}^{n-4}d_{i}^{\m_{n-i-3}})d_{n-2}d_{n-1} &\text{ if }n \geqslant 5.
\end{array}
\right.
\end{equation}
As before, the difference between  the multiplicities of $d_{1}$ and $d_{2}$ is periodic. We have:
\[
\m_{k}^{a}(d_{i})-\m_{k}^{b}(d_{i})=
\left\{
\begin{array}{ll}
w_{k \mod  7} &\text{ if } i=1, \\
w_{k-4 \mod  7} &\text{ if } i=2,
\end{array}
\right.
\] where $w=[1, 0, -1, 0, 1, -1, -1]$, which can be proven by induction as was done in the proof of Lemma 1.
 From this, it follows that in terms of the initial values of the map (\ref{hqrt}), $h_{1}$ and $h_{2}$, we have
\begin{equation}\label{map392}
 \gamma^h_{n}:=\frac{\gamma a_{n}}{d_{n}d_{n-3}g_{n}}=\gamma h_{1}^{r_{n \mod 7}}h_{2}^{r_{n-4 \mod 7}}
\text{ and }
\delta^h_{n}:=\frac{\delta b_{n}}{d_{n-1}d_{n-2}g_{n}}=\delta h_{1}^{s_{n \mod 7}}h_{2}^{s_{n-4 \mod 7}},
\end{equation}
where
\begin{equation}\label{map39}
r_{k}=\text{max}(0,w_{k})=[1,0,0,0,1,0,0]\text{ and }
s_{k}=\text{max}(0,-w_{k})=[0,0,1,0,0,1,1].
\end{equation}
Solving (\ref{map392}) for $a_{n}$ and  $b_{n}$ in terms of $\gamma^h_{n}$ and $\delta^h_{n}$  and substituting into  (\ref{map22}) and (\ref{map23}), we find the following recursion relations.
\begin{theorem}
The sequence $\{d_{n}\}_{n=1}^\infty$, defined by (\ref{X4}), satisfies $d_1=h_1$, $d_2=h_2$, $d_{3}=\gamma h_{2}+\delta$,
\begin{align}\label{MAP1}
\begin{array}{ll}
 d_{4}=\gamma d_{3}+\delta h_{1}h_{2}, \quad d_{5}=\gamma d_{4}+\delta h_{2}d_{3},\quad d_{6}=\gamma h_{1}h_{2}d_{5}+\delta d_{3}d_{4},\quad 
d_{7}=\gamma d_{3}d_{6}+\delta h_{1}d_{4}d_{5},
\end{array}
\end{align}
and, for all $n \geqslant 8$, 
\begin{align}\label{M5}
d_{n-4}d_{n+1}=\gamma_{n}^{h}d_{n}d_{n-3}+\delta_{n}^{h}d_{n-1}d_{n-2}.
\end{align} We note that (\ref{MAP1}) are obtained from (\ref{M5}) by taking initial values $d_{i}=1$ for $i \in \{-2,-1,0,1,2\}$.
\end{theorem}
Again, the fact that $\{d_{n}\}$ is a sequence of polynomials is therefore also explained by the Laurent property of (\ref{M5}), see section \ref{subs}. 

\bigskip

Finally, we'd like to mention that the
third order mapping \cite[Equation 2.9]{Hone5},
\begin{equation}
u_{n+2}u_{n}^{2}u_{n+1}^{2}u_{n-1}=\gamma u_{n}u_{n+1}+\delta,
\end{equation}
which is related to Somos-5 via
\[
u_n=\frac{\sigma_{n+1}\sigma_{n-1}}{\sigma_n^2},
\]
can be recursively factorised as $u_n=a_n/b_n$ with
\begin{equation}\label{map46}
a_{n}=
\left\{
\begin{array}{ll}
d_{n} &\text{ if } n\leqslant 4, \\
d_{1}d_{2}^{2}d_{3}^{3}d_{5} &\text{ if } n=5, \\
d_{1}^{\m_{n}^{a}(d_{1})}d_{2}^{\m_{n}^{a}(d_{2})}d_{3}^{\m_{n}^{a}(d_{3})}(\prod_{i=4}^{n-3}d_{i}^{\m_{n-i}})d_{n-2}^{3}d_{n}, &\text{ if }n > 5,
\end{array}
\right.
\end{equation}
and
\begin{equation}\label{map47}
b_{n}=
\left\{
\begin{array}{ll}
1 &\text{ if } n\leqslant 3 ,\\
d_{n-3}d_{n-2}^{2}d_{n-1}^{2} &\text{ if } n \in \{4,5\}, \\
d_{1}^{\m_{n}^{b}(d_{1})}d_{2}^{\m_{n}^{b}(d_{2})}d_{3}^{\m_{n}^{b}(d_{3})}(\prod_{i=4}^{n-3}d_{i}^{\m_{n-i}})d_{n-2}^{2}d_{n-1}^{2}, &\text{ if }n > 5,
\end{array}
\right.
\end{equation}
where 
\[
\m_{1}=0,\ \m_{2}=3,\ \m_{3}=7,\ \m_{n+2}=2\m_{n+1}+2\m_{n}+\m_{n-1},
\]
and, for $i\in\{1,2,3\}$, $\{\m_{n}^{a}(d_{i})\}$ and $\{\m_{n}^{b}(d_{i})\}$ are defined by initial values $\{\m_{j}^{a}(d_{i})=\delta_{ij},\m_{j}^{b}(d_{i})=0\}_{i,j=1}^3$, and
\begin{align}\label{map44}
 \m_{n+2}^{a} &= min(\m_{n}^{a}+\m_{n+1}^{a} + \m_{n}^{b}+ \m_{n+1}^{b}+\m_{n-1}^{b};2\m_{n}^{b}+2\m_{n+1}^{b} +\m_{n-1}^{b}),\\
\m_{n+2}^{b} &= 2\m_{n}^{a}+ 2\m_{n+1}^{a}+\m_{n-1}^{a}.\label{map45}
\end{align}
Here, the differences between the multiplicities of $d_{1}$, $d_{2}$, and $d_{3}$ are periodic sequences with period 14. We have:
\[
\m_{k}^{a}(d_{i})-\m_{k}^{b}(d_{i})=
\left\{
\begin{array}{ll}
h_{k \mod 14} &\text{ if } i=1, \\
h_{k \mod 14}+h_{k+3 \mod 14} &\text{ if } i=2,\\
h_{k-4 \mod 14}&\text{ if } i=3,
\end{array}
\right.
\] where $h=[1, 0, 0, -1, 1, 0, -1, 0, 1, -1, 0, 0, 1, -2]$, from which it follows that 
$
\phi_{n}:=\frac{ a_{n}}{d_{n}d_{n-2}g_{n}}
$
and
$ 
\psi_{n}:=\frac{ b_{n}}{d_{n-1}^{2}g_{n}}
$
are periodic with period 14. However, the coefficients of the periodic Somos-5 recurrence for the sequence $\{d_n\}$ defined by (\ref{map46}),
\begin{align}\label{map52}
d_{n-3}d_{n+2}=\gamma^{u}_{n+2} d_{n-2}d_{n+1}+\delta^{u}_{n+2} d_{n}d_{n-1},
\end{align}
turn out to have period 7,
\begin{align*}
\gamma^{u}_{n+2}=
&\gamma\frac{\psi_{n-1} \psi_{n} \psi_{n+1}\psi_{n+2}}{\phi_{n-1}\phi_{n}\phi_{n+1}\phi_{n+2}}
=\gamma u_1^{r_{n \mod 7}}u_2^{r_{n \mod 7}+r_{n-3 \mod 7}}u_3^{r_{n-3 \mod 7}},\\
\delta^{u}_{n+2}=&\delta\frac{\psi_{n-1} \psi_{n}^{2} \psi_{n+1}^{2}\psi_{n+2}}{\phi_{n-1}\phi_{n}^{2}\phi_{n+1}^{2}\phi_{n+2}}
=\delta u_1^{s_{n \mod 7}}u_2^{s_{n \mod 7}+r_{n-3 \mod 7}}u_3^{s_{n-3 \mod 7}},
\end{align*}
with $r,s$ as before. 
Thus, equation (\ref{map52}) sits inside the periodic Somos-5 family mentioned in the introduction, equation (\ref{S5P}).


\section{From Somos-4 \& 5 recurrences to Somos-4 \& 5 recurrences with periodic coefficients}
\label{SSOM}
In this section we derive the Somos sequences with periodic coefficients mentioned in the introduction,
which are more general than the ones obtained from QRT-maps in the previous section.
\subsection{Periodic Somos-4}
By taking $\tau_n=a_n/b_n$ in Somos-4 we find the
system of recurrences for polynomial sequences $\{a_n\}$ and $\{b_n\}$:
\begin{align} \label{mab}
a_{n+2}&=w_{n+2}b_{n-2} \\
b_{n+2}&=b_{n+1}b_n^2b_{n-1}a_{n-2},\label{maba}
\end{align}
with $w_{n+2}:=\alpha a_{n+1}b_n^2a_{n-1}+\beta b_{n+1}a_n^2b_{n-1}$. Taking $\{b_i=1\}_{i=1}^4$, we have
$\{a_i=\tau_i\}_{i=1}^4$. From (\ref{mab}) and (\ref{maba}), we find the following ultra-discrete system of recurrences for a lower bound on multiplicities:
\begin{align}\label{m2}
 \m_{n+2}^{a} &= min(\m_{n+1}^{a} + 2\m_{n}^{b}+\m_{n-1}^{a} + \m_{n-2}^{b};\m_{n+1}^{b}+ 2\m_{n}^{a}+\m_{n-1}^{b} +\m_{n-2}^{b}),\\
\m_{n+2}^{b} &= \m_{n+1}^{b}+ 2\m_{n}^{b}+\m_{n-1}^{b} + \m_{n-2}^{a}.\label{m3}
\end{align} 
Iterating the recurrences (\ref{mab}), (\ref{maba}) four times gives us
\begin{align*}
a_{n+3} = b_{n-1} b_{n+1} p_1,
\ &b_{n+3}=a_{n-2} a_{n-1} b_{n-1} b_n^3 b_{n+1}^3 , \\
a_{n+4} = a_{n-2} b_{n-1} b_n^4 b_{n+1}^3 p_2,
\ &b_{n+4} = a_{n-2}^3 a_{n-1} a_{n} b_{n-1}^3 b_{n}^7 b_{n+1}^6, \\
a_{n+5} = a_{n-2}^3 a_{n-1} b_{n-1}^3 b_{n}^9 b_{n+1}^{10}p_3,
\ &b_{n+5} =  a_{n-2}^6 a_{n-1}^3 a_{n} a_{n+1}b_{n-1}^6 b_{n}^{15} b_{n+1}^{13}, \\
a_{n+6} = a_{n-2}^{10} a_{n-1}^3 a_{n} b_{n-1}^{10} b_{n}^{25} b_{n+1}^{23}w_{n+2}p_4,
\ &b_{n+6} = a_{n-2}^{13} a_{n-1}^6 a_{n}^3 a_{n+1} b_{n-2}b_{n-1}^{13} b_{n}^{32} b_{n+1}^{28} w_{n+2},
\end{align*}
where $p_1,p_2,p_3,p_4$ are irreducible polynomials in $\{a_{n+i},b_{n+i}\}_{i=-2}^1$,
$\alpha$ and $\beta$. Thus, considering a lower bound  for the multiplicity of $w_{k}$ $(k>4)$ in the
sequences $\{a_{n}\}$ and $\{b_{n}\}$, we  solve (\ref{m2}) and (\ref{m3}) with  the following initial values:
$\m_{k+i}^{a}=\m_{k+i}^{b}=0$,   where  $ i \in \{1,2,3\}$ and $\m_{k+4}^{a}=\m_{k+4}^{b}=1$. We find, for
$n \geqslant k+1$, 
\[
\m_{n}^{a}(w_{k})= \m_{n}^{b}(w_{k})=\m_{n-k},
\]  where the integer sequence $\{\m_{n}\}$ is defined by 
\begin{align*}
\m_{n+2}=\m_{n+1}+2\m_{n}+\m_{n-1}+\m_{n-2},
\end{align*}
and $\m_{1}=\m_{2}=\m_{3}=\m_{4}-1=0$. We define sequences $\{\m_{n}^{a}(c_{i})\}$ and $\{\m_{n}^{b}(c_{i})\}$, for $i\in\{1,2,3,4\}$, by (\ref{m2}) and (\ref{m3}) and  the initial values $\{\m_{j}^{a}(c_{i})=\delta_{ij},\m_{j}^{b}(c_{i})=0\}_{i,j=1}^4$. Next, polynomials $c_{n}$ are defined as a quotient of $a_{n}$, with $n>4$, by
\begin{equation}\label{kh3}
a_{n}
=(\prod_{i=1}^{4}c_{i}^{\m^{a}_{n}(c_{i})})(\prod_{i=5}^{n-1}c_{i}^{\m_{n-i}})c_{n},
\end{equation}
and $b_{n}$ can be expressed as
\begin{align*}
b_{n}=(\prod_{i=1}^{4}c_{i}^{\m^{b}_{n}(c_{i})})\prod_{i=5}^{n-1}c_{i}^{\m_{n-i}}.
\end{align*}
Note that
\begin{align}\label{kh4}
g_{n}=(\prod_{i=1}^{4}c_{i}^{min(\m^{a}_{n}(c_{i}),\m^{b}_{n}(c_{i}))})\prod_{i=5}^{n-1}c_{i}^{\m_{n-i}} \text{ and } 
\frac{b_{n}}{g_{n}}=\prod_{i=1}^{4}c_{i}^{\m^{b}_{n}(c_{i})-min(\m^{a}_{n}(c_{i}),\m^{b}_{n}(c_{i}))},
\end{align} which shows that Somos-4 possesses the Laurent property.
We next study the multiplicities of the divisors $\{c_i\}_{i=1}^4$. But first, let us define the ultra-discrete
Somos-4 recurrence
\begin{equation}\label{X1}
r_{n+4}=-r_{n}+max(r_{n+3}+r_{n+1},2r_{n+2}),
\end{equation}
for which we take initial values $r_{1}=-1$, $r_{2}=r_{3}=r_{4}=0$, cf. \cite[Example 3.6]{For}.
The second difference of this sequence,
\begin{equation} \label{sd}
x_{k}=r_{k+2}-2r_{k+1}+r_{k},
\end{equation}
is a periodic sequence of order $8$.
This follows by iteration of $x_{k+2}+2x_{k+1}+x_{k}=max(x_{k+1},0)$, which itself is an ultra-discrete
version of the QRT-map (\ref{fqrt}), cf. \cite{Nob}.

\begin{lemma}
For all $1 \leqslant i \leqslant 4$, we have:   
$\m_{n}^{b}(c_{i})- \m_{n}^{a}(c_{i})=r_{n-i+1}
$.
\end{lemma}
\begin{proof}
It is enough to proof the lemma for $i=1$, because $ \m_{n}^{a}(c_{i})= \m_{n-1}^{a}(c_{i-1})$ and  $\m_{n}^{b}(c_{i})=\m_{n-1}^{b}(c_{i-1})$ for $i=\{2,3,4\}$ and $n>1$. For brevity we omit the dependence of $\m_{k}^{a}$ and $\m_{k}^{b}$ on $c_{1}$. From initial values, we see that 
$\m_{n}^{b}- \m_{n}^{a}=r_{n}$ for $1 \leqslant n \leqslant 4$. We suppose the lemma is true for all $n \leqslant k$ and demonstrate it is then also true in the case $n=k+1$. We have
\begin{align}
 \m_{k+1}^{a} &= min(\m_{k}^{a} + 2\m_{k-1}^{b}+\m_{k-2}^{a} + \m_{n-3}^{b};\m_{k}^{b}+ 2\m_{k-1}^{a}+\m_{k-2}^{b} +\m_{n-3}^{b}),\label{kh}\\
\m_{k+1}^{b} &= \m_{k}^{b}+ 2\m_{k-1}^{b}+\m_{k-2}^{b} + \m_{k-3}^{a}\label{khe}.
\end{align}
According to the hypothesis, we may replace $\m_{l}^{a}=\m_{l}^{b}-r_{l}$ in the right hand sides of (\ref{kh}) and (\ref{khe}). This gives 
 \begin{align*}
 \m_{k+1}^{a} &= min(\m_{k}^{b} + 2\m_{k-1}^{b}+\m_{k-2}^{b} + \m_{k-3}^{b}-r_{k}-r_{k-2};\m_{k}^{b}+ 2\m_{k-1}^{b}+\m_{k-2}^{b} +\m_{k-3}^{b}-2r_{k-1})\\
&= \m_{k}^{b} + 2\m_{k-1}^{b}+\m_{k-2}^{b} + \m_{k-3}^{b} +min(-(r_{k}+r_{k-2}),-2r_{k-1})\\
&=\m_{k}^{b} + 2\m_{k-1}^{b}+\m_{k-2}^{b} + \m_{k-3}^{b} - max(r_{k}+r_{k-2},2r_{k-1})\\
\m_{k+1}^{b} &= \m_{k}^{b}+ 2\m_{k-1}^{b}+\m_{k-2}^{b} + \m_{k-3}^{b}-r_{k-3}.
\end{align*} Thus,
$
\m_{k+1}^{b}-\m_{k+1}^{a}=-r_{k-3}+ max((r_{k}+r_{k-2}),2r_{k-1})=r_{k+1}.
$
\end{proof}
\begin{theorem}
For all $n>4$, the polynomials $c_{i}$ defined by (\ref{kh3}) satisfy the Somos-4 recurrence
\begin{align}\label{khe1}
c_{n-2}c_{n+2}=\alpha_{n}^{\tau}c_{n+1}c_{n-1}+\beta_{n}^{\tau}c_{n}^{2},
\end{align} with periodic coefficients,
\begin{equation}
\alpha_{n}^{\tau}=\alpha (\prod_{i=1}^{4}\tau_{i}^{p_{n-i \mod 8}})
\end{equation} 
and 
\begin{equation}
\beta_{n}^{\tau}=\beta (\prod_{i=1}^{4}\tau_{i}^{q_{n-i \mod 8}}),
\end{equation} with  $p$  and $q$ as given in Corollary 2, and
initial values $\{c_{i}=1\}_{i=1}^{4}$.
\end{theorem}
 \begin{proof}
From (\ref{kh3}), (\ref{kh4}) and the initial values we obtain
\begin{align}\label{Somo1}
 c_{5}&=\alpha c_{2}c_{4}- \beta c_{3},\quad c_{6}=\alpha c_{3}c_{5}+\beta c_{1}c_{4}
,\quad c_{7}=\alpha c_{1}c_{4}c_{6}+\beta c_{2}c_{5},\quad \text{ and }
c_{8}&= \alpha c_{1}c_{4}c_{6}+\beta c_{2}c_{5},
\end{align} 
Using Lemma 5, we find $a_{n}=g_{n}c_{n}$ and
$b_{n}=(\prod_{i=1}^{4}c_{i}^{r_{n-i+1}})g_{n}$.
Substituting these expressions in (\ref{mab}) gives, for $n>8$
\[
c_{n+2}g_{n+2}=g_{n-1}g_{n-2}g_{n}^{2}g_{n+1}(\alpha c_{n+1}c_{n-1} \prod_{i=1}^{4}c_{i}^{2r_{n-i+1}+r_{n-i-1}}+\beta c_{n}^{2}\prod_{i=1}^{4}c_{i}^{r_{n-i+2}+r_{n-i-1}+r_{n-i}}).
\]
From (\ref{maba}), we find:
\[
\prod_{i=1}^{4}c_{i}^{r_{n-i+3}}g_{n+2}=g_{n-1}g_{n-2}g_{n}^{2}g_{n+1}(\prod_{i=1}^{4}c_{i}^{r_{n-i+2}+2r_{n-i+1}+r_{n-i}})c_{n-2}.
\]
Eliminating $g_{n+2}$ from the above yields
\[
c_{n+2}c_{n-2}=\alpha(\prod_{i=1}^{4}c_{i}^{r_{n-i+3}-r_{n-i+2}-r_{n-i}+r_{n-i-1}})c_{n+1}c_{n-1}+\beta(\prod_{i=1}^{4}c_{i}^{r_{n-i+3}-2r_{n-i+1}+r_{n-i-1}})c_{n}^{2},
\]
which can be expressed in terms of $p$ and $q$
\[
r_{n-i+3}-r_{n-i+2}-r_{n-i}+r_{n-i-1}=x_{n-i+1}+x_{n-i}+x_{n-i-1}=p_{n-i \mod 8},
\]
and 
\[
r_{n-i+3}-2r_{n-i+1}+r_{n-i-1}=x_{n-i+1}+2x_{n-i}+x_{n-i-1}=q_{n-i \mod 8}.
\]
Taking unit initial values, the relations (\ref{Somo1}) are generated by (\ref{khe1}). 
\end{proof}

\subsection{Periodic Somos-5}
For Somos-5 we follow the same steps. Homogenising $\sigma_{n}=a_{n}/b_{n}$ gives
\begin{align}\label{map3}
a_{n+3}=&w_{n+3}b_{n-2},\\
b_{n+3}=&b_{n+2}b_{n-1}b_{n}b_{n+1}a_{n-2},\label{map4}
\end{align} where $w_{n+3}:=\gamma a_{n+2}a_{n-1}b_{n}b_{n+1}+\delta a_{n}a_{n+1}b_{n+2}b_{n-1}$.
We take $\{b_{i}=1\}_{i=1}^5$ and so $\{a_{i}=\sigma_{i}\}_{i=1}^5$. Iterating (\ref{map3}) and (\ref{map4}) five more times, we find:
\begin{align*}
a_{n+4} = b_{n+2}b_{n+1}b_{n-1}q_{1},
\ &b_{n+4}=a_{n-2}a_{n-1}b_{n-1}b_n^2 b_{n+1}^2b_{n+2}^2, \\
a_{n+5} =a_{n-2}b_{n-1}b_n^2 b_{n+1}^2b_{n+2}^2 q_2,
\ &b_{n+5} = a_{n-2}^2 a_{n-1}a_{n}b_{n-1}^2b_{n}^3 b_{n+1}^4 b_{n+2}^4, \\
a_{n+6} =a_{n-2}^2a_{n-1}b_{n-1}^3b_{n}^3 b_{n+1}^6 b_{n+2}^5 q_3,
\ &b_{n+6} = a_{n-2}^4 a_{n-1}^2a_{n} a_{n+1}b_{n-1}^4b_{n}^{6}b_{n+1}^7  b_{n+2}^{8}, \\
a_{n+7} =a_{n-2}^{5} a_{n-1}^{2}a_{n}b_{n-1}^6 b_{n}^{8} b_{n+1}^{11} b_{n+2}^{12}q_4,
\ &b_{n+7} = a_{n-2}^{8} a_{n-1}^4 a_{n}^2 a_{n+1}a_{n+2}b_{n}^{12} b_{n+1}^{14} b_{n+2}^{15},\\
a_{n+8} =a_{n-2}^{12}a_{n-1}^5 a_{n}^2 a_{n+1}b_{n-1}^{14}b_{n}^{18} b_{n+1}^{24} b_{n+2}^{25} w_{n+3}q_5,
\ &b_{n+8} =  a_{n-2}^{15}a_{n-1}^8 a_{n}^4 a_{n+1}^2 a_{n+2}b_{n-2} b_{n-1}^{15}b_{n}^{23}b_{n+1}^{27}b_{n+2}^{29}w_{n+3},
\end{align*}
where $\{q_{i}\}_{i=1}^{5}$ are irreducible polynomials in $\{a_{n+i},b_{n+i}\}_{i=-2}^2$, $\delta$ and $\gamma$. From (\ref{map3}) and (\ref{map4}), the system that gives a lower bound for multiplicities is:
\begin{align}\label{map5}
 \m_{n+3}^{a} &= min(\m_{n+2}^{a}+\m_{n-1}^{a}+\m_{n}^{b}+\m^{b}_{n+1}+\m^{b}_{n-2};\m_{n}^{a}+\m_{n+1}^{a} +\m_{n+2}^{b}+\m_{n-1}^{b}+\m_{n-2}^{b}),\\
\m_{n+3}^{b} &= \m_{n+2}^{b}+ \m_{n-1}^{b}+\m_{n}^{b}+\m_{n+1}^{b}+ \m_{n-2}^{a}.\label{map6}
\end{align}
To obtain a lower bound for $\m^{a}_{n}(w_{k})$ and $\m^{b}_{n}(w_{k})$,  we solve (\ref{map5}) and (\ref{map6}) with the following initial values: $ \m_{k+i}^{a}=\m_{k+i}^{b}=0$  for all $i\in \{1,2,3,4\}$ and $ \m_{k+5}^{a}=\m_{k+5}^{b}=1$. We find,  for all $n \geqslant k+1,$ 
\begin{align*}
\m_{n}^{a}(w_{k})= \m_{n}^{b}(w_{k})=\m_{n-k},
\end{align*} where $\m_{n+4}=\m_{n+3}+\m_{n+2}+\m_{n+1}+\m_{n}+\m_{n-1}$ and $\m_{1}=\m_{2}=\m_{3}=\m_{4}=\m_{5}-1=0$. Then, the formulae for $a_n$ and $b_n$ in terms of a new sequence $\{d_k\}_{k=1}^\infty$ are given as follows, with $n>5$, 
\begin{equation}\label{LP1}
a_{n}=(\prod_{i=1}^{5}d_{i}^{\m^{a}_{n}(d_{i})})(\prod_{i=6}^{n-1}d_{i}^{\m_{n-i}})d_{n},\quad 
b_{n}=(\prod_{i=1}^{5}d_{i}^{\m^{b}_{n}(d_{i})})\prod_{i=6}^{n-1}d_{i}^{\m_{n-i}},
\end{equation}
where sequences $\{\m_{n}^{a}(d_{i\leq5})\}$ and $\{\m_{n}^{b}(d_{i\leq5})\}$ are defined by (\ref{map5}) and (\ref{map6}) and  the initial values $\{\m_{j}^{a}(d_{i})=\delta_{ij},\m_{j}^{b}(d_{i})=0\}_{i,j=1}^5$. These formulae clearly show that Somos-5 possesses the Laurent property.
The differences between the multiplicities of $c_{i\leq5}$ can be expressed in terms of the ultra-discrete
Somos-5 sequence defined by
\begin{equation}\label{XXX}
t_{n+5}=-t_{n}+max(t_{n+4}+t_{n+1},t_{n+3}+t_{n+2}),
\end{equation}
and initial values $t_{1}=-1$, $\{t_{i}=0\}_{i=2}^{5}$. The quantity
\begin{align}\label{XX1}
y_{k+1}=t_{k+3}-t_{k+2}-t_{k+1}+t_{k},
\end{align}
which relates to (\ref{hs}), satisfies the ultra-discrete QRT-map (related to \ref{hqrt}),
\begin{align*}
y_{k+2}+y_{k+1}+y_{k}=max(y_{k+1},0),
\end{align*} and is periodic of order 7. It follows from (\ref{map5}) and (\ref{map6})
that
\begin{equation} \label{mds}
\m_{n}^{b}(d_{i})- \m_{n}^{a}(d_{i})=t_{n-i+1},
\end{equation}
for all $1 \leqslant i \leqslant 5$.
\begin{theorem} 
Let $r,s$ be given as in (\ref{map39}). For all  $n >7$, the polynomials $d_{n>5}$ defined by (\ref{LP1}), satisfy the Somos-5 recurrence with periodic coefficients 
\begin{equation} \label{PS5T}
d_{n+3}d_{n-2}=\gamma_{n}
d_{n+2}d_{n-1}+\delta_{n}
d_{n}d_{n+1},
\end{equation} where
\begin{equation}
\gamma_n
=\gamma(\prod_{i=1}^{5}\sigma_{i}^{r_{n-i \mod 7}}),\quad
\delta_n
=\delta (\prod_{i=1}^{5}\sigma_{i}^{s_{n-i \mod 7}}),
\end{equation}
and initial values $\{d_i=1\}_{i=1}^5$. 
\end{theorem}
\begin{proof}
Using (\ref{map3}), (\ref{map4}), initial values and (\ref{LP1}), we find 
\begin{align}\label{Pet1}
d_{6}&=\alpha \sigma_{5}\sigma_{2}+\beta \sigma_{3}\sigma_{4},\quad
d_{7}=\alpha d_{6}\sigma_{3}+\beta \sigma_{4}\sigma_{5}\sigma_{1},\quad 
d_{8}=\alpha d_{7}\sigma_{4}+\beta \sigma_{2}d_{6}\sigma_{5},\quad \\\notag
d_{9}&=\alpha d_{8}\sigma_{5}\sigma_{1}+\beta \sigma_{3}d_{7}d_{6},\quad \text{ and }
d_{10}=\alpha d_{9}d_{6}\sigma_{2}+\beta \sigma_{1}\sigma_{4}d_{8}d_{7}.
\end{align}
For all $n>10$, from (\ref{LP1}) and (\ref{mds}), we find
$a_{n}=g_{n}d_{n}$ and $b_{n}=(\prod_{i=1}^{5}d_{i}^{t_{n-i+1}})g_{n}$.
Substitution in equation (\ref{map3}) gives
\begin{align*}
d_{n+3}g_{n+3}=g_{n-2}g_{n-1}g_{n}g_{n+1}g_{n+2}(&\gamma d_{n+2}c_{n-1}\prod_{i=1}^{5}d_{i}^{t_{n-i+1}+t_{n-i+2}+t_{n-i-1}}\\
&+\delta d_{n}d_{n+1}\prod_{i=1}^{5}d_{i}^{t_{n-i+3}+t_{n-i}+t_{n-i-1}}).
\end{align*}
From (\ref{map4}), we find:
\[
g_{n+3}\prod_{i=1}^{5}d_{i}^{t_{n-i+4}}=g_{n-2}g_{n-1}g_{n}g_{n+1}g_{n+2}(\prod_{i=1}^{5}d_{i}^{t_{n-i+3}+t_{n-i}+t_{n-i+1}+t_{n-i+2}})d_{n-2}.
\]
Eliminating $g_{n+3}$ from the above two equations yields
\begin{align*}
d_{n+3}d_{n-2}=&\gamma(\prod_{i=1}^{5}d_{i}^{t_{n-i+4}-t_{n-i+3}-t_{n-i}+t_{n-i-1}})d_{n+2}d_{n-1}\\
&+\delta (\prod_{i=1}^{5}d_{i}^{t_{n-i+4}-t_{n-i+1}-t_{n-i+2}+t_{n-i-1}})d_{n}d_{n+1},
\end{align*}
which can be expressed in terms of $r$ and $s$ as follows
\begin{align*}
t_{n-i+4}-t_{n-i+3}-t_{n-i}+t_{n-i-1}=y_{n-i+1}+y_{n-i-1}=r_{n-i \mod 7},
\end{align*}
and 
\begin{align*}
t_{n-i+4}-t_{n-i+1}-t_{n-i+2}+t_{n-i-1}=y_{n-i+1}+y_{n-i}+y_{n-i-1}=s
_{n-i \mod 7}.
\end{align*}
Taking $\{d_{i}=1\}_{1}^{5}$ , then (\ref{Pet1}) are generated by (\ref{PS5T}).
\end{proof}

\subsection{On the Laurent property of periodic Somos-4 \& 5 sequences} \label{subs}
As the periodic Somos-4 \& 5 sequences we have derived are special cases of
equation (\ref{HWE}) and the condition (\ref{HWC}) is satisfied, they possess the Laurent
property. 

If we would not have had the Hirota-Miwa equation at hand, or one wants
a direct proof this can be done. Actually, 
most of the work has been done already. Considering (\ref{khe1}),  the substitution $c_n=a_n/b_n$  yields the same system of recurrences (\ref{mab}), (\ref{maba}) for polynomials $a_n$ and $b_n$. The only difference is that in the expression for $w_{n+2}$, the $\alpha$ and $\beta$ are now periodic functions of $n$ with period 8.
This means that the iteration of the recurrences (four more times) has to be repeated for different values of $n\equiv i$ mod 8, with $i\in\{0,1,2,\ldots,7\}$. For each value of $i$ we found that $w_{n+2}$ does not divide $a_{n+k}$ or $b_{n+k}$, with $k=3,4,5$, and that it does divide both $a_{n+6}$ and $b_{n+6}$. As the system of recurrences is similar, the derived ultra-discrete system (\ref{m2}), (\ref{m3}) is the same, polynomials $c_n$ are defined by equation (\ref{kh3}), and the proof carries over. Also no surprises were found when iterating the system (\ref{map22},\ref{map23}) five more times, for $p=7$ different values for $n$ mod $p$.

\section{From DTKQ equations to sequences of polynomials which satisfy rational recurrences}  
The aim of this section is to study how the second and third order DTKQ equations give rise to recurrences that possess the Laurent property. The $N$th order DTKQ equation
\begin{equation}\label{D1}
\sum_{s=0}^{N}u_{n+s}\prod_{q=1}^{N-1}u_{n+q}=\alpha. 
\end{equation}
was derived in \cite{DTKQ}, through applying the principle of duality for difference equations, and it was shown
to admit sufficiently many integrals to be completely integrable. The growth of the equations has been studied in \cite{Ham}.

\subsection{From the second order DTKQ equation to a fifth order Laurent recurrence with four terms.}
In the case  $N=2$, the DTKQ equation is 
\begin{equation} \label{1}
u_{n+2} = \frac{\alpha}{u_{n+1}}-u_{n}-u_{n+1}.
\end{equation}
Substituting  $u_{n}=a_{n}/b_{n}$ in (\ref{1}) and identifying the numerators and denominators, we get a system of recurrences for polynomial sequences $\{a_{n}\}$ and $\{b_{n}\}$:
\begin{align}\label{1a}
 a_{n+2} &=\alpha b_{n}b_{n+1}^2-a_{n}a_{n+1}b_{n+1}-b_{n}a_{n+1}^{2},\\
 b_{n+2} &= a_{n+1}b_{n}b_{n+1},\label{1b}
\end{align} with $a_{1}=u_{1}$,  $a_{2}=u_{2}$, $b_{1}=b_2=1$.
Therefore, $a_{n}$ and $b_{n}$ are polynomials in the variables  $u_{1}$ and $u_{2}$. Via recursive factorisation, we find, see \cite{Ham}, in terms of a polynomial sequence $\{e_n\}$,
\begin{align} \label{1c}
a_{n}&=
\left\{
\begin{array}{ll}
e_{n}   & \text{ if } n\leqslant 3, \\
e_{n}e_{n-3}\prod_{i=2}^{n-3}e_{i}^{\m_{n-i-2}} & \text{ if } n > 3,
\end{array}
\right. \\
b _{n} &=\left\{
\begin{array}{ll}
1   & \text{ if } n\leqslant 2, \\
e_{2}   &  \text{ if } n=3, \\
e_{n-1}e_{n-2}\prod_{i=2}^{n-3}e_{i}^{\m_{n-i-2}} & \text{ if } n > 3,
\end{array} \right.\\
\end{align}
with $\m_{1}=2$, $\m_{2}=6$, and $\m_{l}=2\m_{l-1}+\m_{l-2}$.
\begin{theorem}
The polynomials $e_{n>3}$ satisfy
\begin{align}\label{a3}
\frac{e_{n-1}e_{n-5}}{e_{n-3}^{2}} + \frac{e_{n-1}^{2}e_{n-4}^{2}}{e_{n-3}^{2}e_{n-2}^{2}} + \frac{e_{n-4}e_{n}}{e_{n-2}^{2}} = \alpha,
\end{align}
where 
\begin{equation} \label{inva}
\{e_{i}=1\}_{i=-1}^{1},\quad e_{2}=u_{2} \quad \text { and } \quad e_{3}=\alpha - u_{1}u_{2}-u_{2}^{2}.
\end{equation}
\end{theorem}
We remark that a reduction of order, by introducing the variable $e_{n+1}/e_{n-1}$, is apparent, however, this does not preserve Laurentness. Furthermore we should mention that
the fact that the rational recurrence (\ref{a3}) with initial values (\ref{inva}) produces
a polynomial sequence does not follow from the Laurent property of (\ref{a3}). One needs a
strong Laurent property such as given in \cite{HS} for Somos sequences.
 
\begin{proof}
From (\ref{1c}), (\ref{1a}), (\ref{1b}) and  initial values, we obtain
\begin{align*}
e_{2}=a_{2}, \quad e_{3}=a_{3}=\alpha b_{1}{b_{2}}^2-a_{1}a_{2}b_{2}-b_{1}{a_{2}}^{2} = \alpha -u_{1}u_{2}-u_{2}^{2}.
\end{align*}
Similarly, we find
\begin{align}
e_{4}&=a_{4}
= \alpha u_{2}^{2} -e_{3}u_{2}^{2}-e_{3}^{3}, \label{c4}\\
e_{5}&=\frac{a_{5}}{e_{2}g_{5}}=\frac{ \alpha u_{2}^{2}e_{3}^{2}-e_{4}e_{3}^{2}-e_{4}^{2}}{u_{2}^{2}},\label{c5}\\
e_{6}&=\frac{a_{6}}{e_{3}g_{6}}=\frac{\alpha e_{3}^{2}e_{4}^{2}-e_{5}e_{4}^{2} -e_{5}^{2}u_{2}^{2}}{u_{2}e_{3}^{2}}.\label{c6}
\end{align}
Now consider, for $n >4$, replacing $a_{n+i}$  by $ e_{n+i}c_{n-3+i}g_{n+i}$ and  $b_{n+i}$  by $ e_{n-1+i}e_{n-2+i}g_{n+i}$  in the right hand side of equation (\ref{1a}):
\[
e_{n+2}e_{n-1}g_{n+2}
=g_{n}g_{n+1}^{2}e_{n-1}e_{n-2}(\alpha e_{n}^{2}e_{n-1}^{2}-e_{n}^{2}e_{n-3}e_{n+1}-e_{n+1}^{2}e_{n-2}^{2}).
\]
%
From equation (\ref{1b}) we find
$
g_{n}g_{n+1}^{2}=\dfrac{g_{n+2}}{e_{n-1}^{2}e_{n-2}^{2}}
$
and these combine to give the recurrence equation for $e's$, (\ref{a3}).
By taking $e_{-1}=e_{0}=e_{1}=1$, the recurrence equation (\ref{a3}) generates the above expressions (\ref{c4}), (\ref{c5}) and (\ref{c6})
\end{proof}

We could now recursively factorize the equation (\ref{a3}), but if
one just wants to verify the Laurent property there is an easier method,
as described in \cite{Rob}. By iterating the map five times we obtain
$\{e_n=p_n/q_n\}_{n=5}^{10}$, for polynomials $p_n$ and monomials $q_n$ in the initial values $\{e_n\}_{n=1}^5$. As $p_{5}$ is prime to $p_{n}$ for all $n \in \{6,7,8,9,10\}$ the recurrence (\ref{a3}) satisfies the Laurent property. 

\subsection{From the third order DTKQ equation to a sixth order Laurent recurrence with five terms, with coefficients that are periodic with period 8.}
Taking $N=3$ in equation (\ref{D1}), this gives the third order DTKQ equation, 
 \begin{align*}
   u_{n+3} = \dfrac{\alpha}{u_{n+1}u_{n+2}} - u_{n}  - u_{n+1} - u_{n+2},
  \end{align*} and  homogenising yields,
\begin{align}\label{equ1}
a _{n+3}&= \alpha b_{n+1}^{2}b_{n+2}^{2}b_{n} - a_{n+1}a_{n+2}a_{n}b_{n+1}b_{n+2} - a_{n+1}^{2}a_{n+2}b_{n+2}b_{n} - a_{n+1}a_{n+1}^{2}b_{n+1}b_{n},\\                                            
 b _{n+3}& = b _{n+1}b _{n+2}b _{n}a_{n+1}a_{n+2}\label{equ2}.
 \end{align} If we choose
$\{a_{i}=u_{i},b _{i} = 1\}_{i=1}^3$ then $a _{l}$ and $b _{l} $ are polynomials in the initial variables $u_{1}$, $u_{2}$, $u_{3}$ and parameter $ \alpha $.
A sequence of polynomials $\{z_{l}\}$ is defined by:
 \begin{align}
  a_{n}&=
\left\{
\begin{array}{ll}
z_{n}   & \text{ if } n<  5, \\
z_{3}z_{5}  & \text{ if } n=  5, \\
 z_{2}^{m_{n}^{a}(z_{2})}\prod_{i=3}^{n-3}z_{i}^{\m_{n-i-2}}z_{l-3}z_{l-2}z_{l}
&\text{ if } n > 5,
\end{array}
\right.\label{equ4} \\
\quad b_{n} &= \left\{
\begin{array}{ll}
1   & \text{ if } n< 4, \\
z_{2}z_{3}   & \text{ if } n=4, \\
z_{2}^{m_{l}^{b}(z_{2})}\prod_{i=3}^{n-3}z_{i}^{\m_{n-i-2}}z_{n-2}^{2}z_{n-1}\ \ \ 
& \text{ if } n > 4,
\end{array}
\right.\label{equ5} 
\end{align}
where $\m$ is the integer sequence defined by $\m_{1}=4$, $\m_{2}=13$, $\m_{3}=37$  and
$
\m_{n}=2\m_{n-1}+2\m_{n-2}+\m_{n-3}$.
In this case the ultra-discretisation of the homogenised system does not give us a sharp bound on the multiplicities $m_{n}^{a}(z_{2})$ and $m_{n}^{b}(z_{2})$. By using prime numbers as initial values
we were able to iterate the map  (\ref{equ1},\ref{equ2}) a little further than usual and thus observe the following. 
\begin{conjecture}
The difference of the multiplicities of $z_2$ in $a_n$ and $b_n$ is periodic, we have
 $m_{n}^{a}(z_{2})-m_{n}^{b}(z_{2})=\zeta_{n \mod 8}$, with  $\zeta=[0,1,0,-1,-1,2,-1,-1]$.
 \end{conjecture} 
Assuming the conjecture, from (\ref{equ4},\ref{equ5}), it follows that
\[
\frac{ a_{n}}{c_{n-3}c_{n}g_{n}}= z_{2}^{max(0,\zeta_{n \mod 8})}
\quad \text{ and } \quad
 \frac{b_{n}}{c_{n-2}c_{n-1}g_{n}}=z_{2}^{max(0,-\zeta_{n \mod 8})}
\]
are polynomial sequences in $z_{2}$. Using these functions we find the following theorem
\begin{theorem}
The polynomials $z_{n>4}$, as defined by (\ref{equ1}), satisfy
\begin{equation} \label{uio}
\epsilon_{n}\frac{z_{n-3}z_{n+1}}{z_{n-1}^2}
+\epsilon_{n+1}\frac{z_{n-2}^{2}z_{n+1}^{2}}{z_{n-1}^2z_{n}^2}
+\epsilon_{n+2}\frac{z_{n-2}z_{n+2}}{z_{n}^2}+
\epsilon_{n+3}\frac{z_{n-2}z_{n+3}}{z_{n-1}z_{n+2}}
=\frac{\alpha}{\epsilon_{n+1}\epsilon_{n+2}}\frac{z_{n}z_{n+1}}{z_{n-1}z_{n+2}},
\end{equation}
with
$\epsilon=u_{2}^{\zeta_{n \mod 8}} $, $\{z_{n}=1\}_{n=-1}^2$, $z_{3}=u_{3}$ and $z_{4}=\alpha-u_{2}u_{3}(u_{1}+u_{2}+u_{3})$.
\end{theorem}
The Laurentness of (\ref{uio}) can be verified as before, this time the
iteration has to be repeated for every congruence class $n$ mod 8.



\section*{Acknowledgments}
This research was supported by the Australian Research Council. Both authors acknowledge useful discussions with Reinout Quispel, and would like to thank
Ralph Willox for bringing to our attention reference \cite{Mas}.

\end{document}